\documentclass[12pt,a4paper]{article}
\usepackage[utf8]{inputenc}
\usepackage[T1]{fontenc}
\usepackage{amsmath,amssymb,enumitem}
\usepackage{color}
\usepackage[usenames,dvipsnames,svgnames,table]{xcolor}
\usepackage{hyperref}
\usepackage{amsthm}
\usepackage{graphicx}

\newtheorem{theorem}{Theorem}[section]

\newtheorem{assumption}[theorem]{Assumption}
\newtheorem{corollary}[theorem]{Corollary}
\newtheorem{definition}[theorem]{Definition}

\newtheorem{lemma}[theorem]{Lemma}
\newtheorem{proposition}[theorem]{Proposition}
\newtheorem{remark}[theorem]{Remark}

\newcommand \R{{\mathbb R}}
\newcommand \esp{{\mathbb E}}
\newcommand \Prob{{\mathbb P}}

\newcommand \maj{>}
\newcommand \mino{<}

\newcommand \til{~}

\newcommand{\be}{\begin{equation}}
\newcommand{\ee}{\end{equation}}

\newcommand{\sigDup}{\sigma_{\mathrm{Dup}}}
\newcommand{\sigBS}{\sigma_{\mathrm{BS}}}
\newcommand{\overg}{\overline g}
\newcommand \dd{{\rm d}}

\title{On the harmonic mean representation of the implied volatility}

\author{Stefano De Marco\footnote{Centre de Math\'ematiques Appliqu\'ees (CMAP), CNRS, Ecole Polytechnique, Institut Polytechnique de Paris, France. \url{stefano.de-marco@polytechnique.edu}
This research has benefited from the financial support of the \emph{Chaire Risques Financiers} (Fondation du Risque)
and the  \emph{Chaire Stress Test, RISK Management and Financial Steering} (Fondation de l'Ecole Polytechnique).
I would like to thank Claude Martini for stimulating discussions.
}}
\date{\today}

\begin{document}

\maketitle

\begin{abstract}
\noindent 
It is well know that, in the short maturity limit, the implied volatility approaches the integral harmonic mean of the local volatility with respect to log-strike, see [Berestycki et al., Asymptotics and calibration of local volatility models, Quantitative Finance, 2, 2002].
This paper is dedicated to a complementary model-free result:\til an arbitrage-free implied volatility in fact is the harmonic mean of a positive function for \emph{any} fixed maturity.
We investigate the latter function, which is tightly linked to Fukasawa's invertible map $f_{1/2}$ [Fukasawa, The normalizing transformation of the implied volatility smile, Mathematical Finance, 22, 2012], and its relation with the local volatility surface.
It turns out that  the log-strike transformation $z = f_{1/2}(k)$ defines a new coordinate system in which the short-dated implied volatility approaches the \emph{arithmetic} (as opposed to harmonic) mean of the local volatility.
As an illustration, we consider the case of the SSVI parameterization: in this setting, we obtain an explicit formula for the volatility swap from options on realized variance.
%in a setting where the SSVI model is used to calibrate options on variance.
\end{abstract}

\section{Introduction}

A classical result in the literature on the implied volatility surface, usually referred to as Berestycki, Busca and Florent's (BBF) formula \cite{BBF}, states that the implied volatility $\sigBS$ generated by a local volatility model
\[
\dd S_t = (r - q) S_t \dd t + \sigma_{\mathrm{loc}}(t, S_t) S_t \, \dd W_t
\]
converges to the harmonic mean of the local volatility in the small maturity limit:
\be \label{e:BBF_formula}
\lim_{T \to 0} \sigBS(T,k)
=
\frac 1{\frac 1k \int_0^k \frac {\dd y}{\sigma_{\mathrm{loc}}(0, S_0 e^y)}}
\qquad \quad \forall \, k,
\ee
where $k = \log \frac K{F^T} \to \log \frac K{S_0}$ denotes log-forward moneyness.
For the asymptotic result \eqref{e:BBF_formula} to hold, the local volatility surface needs to have a well-behaved limit $\sigma_{\mathrm{loc}}(0, \cdot)$ when time tends to zero; see Berestycki et al.\til\cite[Assumption (7)]{BBF} for the precise conditions.
The limit \eqref{e:BBF_formula} can also be obtained using small-time large deviations theory; in this setting, the function $\bigl(\int_0^k \frac {\dd y}{\sigma_{\mathrm{loc}}(0, S_0 e^y)}\bigr)^2$ stems from the finite-dimensional rate function of the process $\log \frac{S_t}{S_0}$.

In this work, we show that a representation analogous to \eqref{e:BBF_formula} actually holds for \emph{every}  fixed maturity and in a model-free setting, that is: as soon as $\sigBS$ is a arbitrage-free implied volatility surface, 
\be \label{e:harmonic_mean_intro}
\sigBS(T,k) 
= \frac 1{\frac 1k \int_0^k \frac {\dd y}{\Sigma(T, y)}}
\qquad \quad \forall \, k, \ \forall \, T,
\ee
for some positive function $\Sigma(T,k)$ (which cannot be interpreted as the local volatility anymore).
We inspect the representation \eqref{e:harmonic_mean_intro} from the point of view of static no-arbitrage conditions, and	 investigate the link of the function $\Sigma$ with the local volatility surface associated to $\sigBS$ via Dupire's formula: as expected, the two functions can be identified in the small time limit (but are different otherwise).

The key element in the (simple) proof of \eqref{e:harmonic_mean_intro} is Fukasawa's seminal work \cite{Fukasawa2012} on the strict monotonicity of the time-dependent Black-Scholes maps
$k \mapsto d_0(k) = - \frac k{\sqrt T \sigBS(T,k)} - \frac{\sqrt T \sigBS(T,k)}2$ and 
$k \mapsto d_1(k) = - \frac k{\sqrt T \sigBS(T,k)} + \frac{\sqrt T \sigBS(T,k)}2$.
%$k \mapsto d_i(k) = - \frac k{\sqrt T \sigBS(T,k)} - (1 - 2 i) \frac{\sqrt T \sigBS(T,k)}2$, $i \in \{0,1\}$,
Since the function $\Sigma$ is linked to the interpolated map $f_{1/2}(k) = -\frac12 (d_0(k) + d_1(k))$, we explore the consequences of the
log-strike transformation $z = f_{1/2}(k)$ (referred to as ``normalizing transformation'' in \cite{Fukasawa2012}) on the geometry of the implied volatility surface.
It turns out that such a change of
%strike
variable transforms the harmonic mean representation  \eqref{e:harmonic_mean_intro} into an \emph{arithmetic} mean representation
\[
\sigma_{1/2}(T, z)
%= A \bigl(\Sigma_{1/2}(T,\cdot) \bigr)(k)
= \frac 1 z \int_0^z \Sigma_{1/2}(T, y) \dd y
\qquad \quad \forall \, z, \ \forall \, T,
\]
where $\sigma_{1/2}$ and $\Sigma_{1/2}$ are, respectively, the implied volatility and its harmonic mean counterpart in the new coordinate system, that is $\sigma_{1/2}(T,f_{1/2}(k)) = \sigBS(T,k)$ and $\Sigma_{1/2}(T,f_{1/2}(k)) = \Sigma(T,k)$, for all $k$.

As an application we show that, under some reasonable conditions on the volatility surface $\sigBS$ one starts from,  the short-time limit of the function
$\Sigma_{1/2}$ can also be identified with the short-time limit of Dupire's local volatility, 
showing that BBF asymptotic formula \eqref{e:BBF_formula} is replaced by an arithmetic mean formula in the new coordinate system.
%in the new strike geometry.
We refer to Theorem \ref{t:arithm_mean_formula} for precise statements.
\medskip

\textbf{Notation and basic definitions}.
We denote $c_{\mathrm{BS}}: (k,v) \in \R \times [0,\infty) \to \R_+ := [0,\infty)$ the normalized Black-Scholes call price with forward log-moneyness $k$ and total implied volatility parameter $v = \sqrt \tau \sigma$ :
\be \label{e:normalized_BS_price}
c_{\mathrm{BS}}(k, v) =
\left \{
\begin{array}{ll}
N \left(d_1(k,v) \right) - e^k N \left(d_0(k,v) \right)
& \mbox{if } v \maj 0
\\
(1 - e^k)^+
& \mbox{if } v = 0
\end{array}
\right.
\ee
where $N$ is the standard Gaussian cdf $N(x) = \int_{-\infty}^x \phi(y) \dd y$, $\phi(y) = \frac 1{\sqrt{2 \pi}} e^{-\frac{y^2}2}$, and
$d_i(k,v) = -\frac{k}{v} - (1 -2 i) \frac v 2$,
$i \in \{0,1\}$,
for every $v \maj 0$.

\section{The harmonic mean representation of the implied volatility} \label{s:harmonic_repr}

Let us recall the following

\begin{definition}[Arbitrage-free implied volatility, fixed maturity] \label{def:arb_free_impl_vol}
Let time to maturity $T$ be fixed. 
We say that a function $v: \R \to \R_+$ is a total implied volatility free of static arbitrage if the function 
\be \label{e:call_price}
K \in (0,\infty) \mapsto C(K) := c_{\mathrm{BS}} \biggl(\log \frac K F, v\Bigl(\log \frac K F \Bigr) \biggr)
\ee
is convex and satisfies $\lim_{K \to \infty} C(K) = 0$ (for some, hence for any, $F \maj 0$).
\\
If $v$ is an arbitrage-free total implied volatility for time to maturity $T$, we denote $\sigBS(k) = \frac 1 {\sqrt T} \, v(k)$ the related implied volatility \emph{tout court}.
\end{definition}

It is well know 
that the conditions on the function $C$ in Definition \ref{def:arb_free_impl_vol} (which are usually referred to as no--butterfly arbitrage conditions) are equivalent to the existence of a pricing measure: if $v$ satisfies Definition \ref{def:arb_free_impl_vol}, there exists a non-negative random variable $X$ with $\esp[X] = 1$ such that
\be \label{e:def_total_impl_vol}
c_{\mathrm{BS}}(k, v(k)) = \esp\left[ (X - e^k)^+ \right],
\qquad
\forall \, k \in \R.	
\ee
Arbitrage-free implied volatilities can  fail to be everywhere differentiable and can vanish on some interval; the regularity and the support of $v$ can of course be linked with the regularity and the support of the law of the random variable $X$ in \eqref{e:def_total_impl_vol}, see  \cite[Lemma 5.2]{RogTeh}.
We restrict our analysis to total implied volatilities that are strictly positive and differentiable:

\begin{assumption} \label{a:assumption_v}
We assume $v \in C^1(\R)$ and $v(k) \maj 0$, for every $k \in \R$.
\end{assumption}

Denote 
\be \label{e:f12}
\begin{aligned}
&
f_0(k) := -d_0(k,v(k)) = \frac{k}{v(k)} + \frac{v(k)} 2,
\\
&f_1(k) := -d_1(k,v(k))  = \frac{k}{v(k)} - \frac {v(k)} 2.
\end{aligned}
\ee
Fukasawa \cite{Fukasawa2012} proved the following result.

\begin{theorem}[Fukasawa \cite{Fukasawa2012}] \label{t:monotonicityf12}
If $v$ is an arbitrage-free total implied volatility satisfying Assumption \ref{a:assumption_v}, then
\[
\frac d{dk} f_i(k) \maj 0
\qquad \forall \, k \in \R,
\quad i \in \{0,1\}.
\]
%where the two maps $f_0$ and $f_1$ are defined in \eqref{e:f12}.
\end{theorem}

The strict monotonicity of the maps $f_0$ and $f_1$ can be exploited to rigorously justify some remarkable model-free pricing formulas for European claims such as the log-contract, see \cite{ChrissMorok99, Fukasawa2012, DM_CM_MGFs} and section \ref{s:pricing_formula} below, and can also be used as a partial characterization of the static no-arbitrage condition on $v$, see Remark \ref{r:no_arb_condition} below and the work carried out in \cite{NoArbSVI}.

The current section is devoted to the following result:
 
\begin{theorem} \label{t:harmonic_repr}
Let $v$ be an arbitrage-free total implied volatility satisfying Assumption \ref{a:assumption_v}.
Then, there exists a unique strictly positive function $h \in C^0(\R)$ such that $v$ is the harmonic mean of $h$:
\be \label{e:harmonic_mean}
v(k) = \frac 1{\frac 1k \int_0^k \frac 1{h(y)} \dd y}
\qquad \forall \, k \neq 0,
\ee
and $v(0) = h(0)$.
\end{theorem}

\begin{proof}
Let $f_{1/2}(k) = \frac k{v(k)}$, so that $f_{1/2} \in C^1(\R)$.
Since, by definition of $f_0$ and $f_1$,
\[
f_{1/2}(k) = \frac 12 \left(f_0(k) + f_1(k) \right),
\]
it follows from Theorem \ref{t:monotonicityf12} that $f_{1/2}'(k) \maj 0$ for every $k \in \R$.
Set $h := \frac 1 {f_{1/2}'}$: the function $h$ is strictly positive, continuous, and such that $\frac k {v(k)} = f_{1/2}(k) = f_{1/2}(0) + \int_0^k \frac 1{h(y)} \dd y = \int_0^k \frac 1{h(y)} \dd y$.
Then by construction, $v$ satisfies equation \eqref{e:harmonic_mean} for every $k \neq0$.
Taking the limit as $k \to 0$ in \eqref{e:harmonic_mean} and using the continuity of $h$,
%$\frac 1k \int_0^k \frac 1{h(y)} dy \to \frac 1{h(0)}$ by continuity of $h$
we obtain that $v$ and $h$ coincide at $k=0$.

The uniqueness of $h$ follows from  \eqref{e:harmonic_mean}: taking derivatives on both sides of  $\int_0^k \frac 1{h(y)} \dd y = \frac k{v(k)}$, we get
\be \label{e:derivative_H}
\frac 1{h(k)} =
\frac{\dd}{\dd k } \frac k{v(k)}
%=
%\frac 1{v(k)} - k \frac{v'(k)}{v(k)^2}
= \frac 1{v(k)}\left(1 - k \frac{v'(k)}{v(k)} \right)
\qquad \forall k \in \R.
\ee
which identifies uniquely the function $h$.
\end{proof}

In view of Theorem \ref{t:harmonic_repr}, it seems reasonable to wonder what is the class of functions having an harmonic mean representation as \eqref{e:harmonic_mean}.
%(if this class coincided with the whole set of $C^1$ and positive functions, Theorem \ref{t:harmonic_repr} would be trivial).
Actually, every strictly positive function $v \in C^1(\R)$ admits the representation \eqref{e:harmonic_mean} for a uniquely determined function $h \in C^0(\R)$: simply define $h$ from \eqref{e:derivative_H}.
The important part of the statement of Thm \ref{t:harmonic_repr} is the \emph{positivity} of $h$: if we start from any function $v$, $h$ will not be positive in general.
In this respect, Theorem \ref{t:harmonic_repr} provides a necessary condition for arbitrage freeness of a total implied volatility $v$, but this condition is unfortunately not sufficient,
as we discuss below.\footnote{If Theorem \ref{t:harmonic_repr} were ``if and only if'', we could generate implied volatilities parameterizations by simply taking harmonic means of positive functions.}

\begin{remark} \label{r:no_arb_condition} 
Consider a strictly positive function $v$: we can always define the maps $f_0$ and $f_1$ as in \eqref{e:f12}.
\begin{enumerate}
%\item In general, the function $h$ needs not be positive: Theorem  \ref{t:harmonic_repr} says that this is the case if $v$ is an arbitrage-free implied volatility.
\item[i)] If we assume $v \in C^2(\R)$, then a computation involving the derivatives of the Black-Scholes call price \eqref{e:normalized_BS_price} yields 
\be \label{e:convexity_call_price}
\frac{\dd^2}{\dd K^2} \,
%c_{\mathrm{BS}}\bigl(k, v(k)\bigr) \big|_{k = \log \frac K F}
c_{\mathrm{BS}}\biggl(\log \frac K {F}, v\Bigl(\log \frac K {F} \Bigr)\biggr)
= 
\frac 1{ F K} \phi(f_0(k)) 
\Bigl( v''(k) + v(k) f_0'(k) f_1' (k)
\Bigr) \Big|_{k = \log \frac K F} 
\ee
(see for example \cite{Fukasawa2012} for the derivation of the expression on right hand side), which shows that the call price function $K \mapsto C(K)$ in \eqref{e:call_price} is convex if and only if
\be \label{e:no_butterfly_arb}
v'' + v \, f_0' \, f_1' \ge 0.
\ee
In particular, we see that the strict monotonicity of $f_0$ and $f_1$ in Theorem \ref{t:monotonicityf12} is a necessary condition for the convexity of $C(\cdot)$, but not a sufficient condition.\footnote{For example, the strict monotonicity of $f_0$ and $f_1$ together with convexity of $v$ would be a sufficient condition. Unfortunately, this condition would be very restrictive in practice, since implied volatility smiles calibrated to market data are often not convex.} 

\item[ii)]  We can further define the function $h$ from equation \eqref{e:derivative_H}, so that $v$ and $h$ satisfy \eqref{e:harmonic_mean} by construction.
Since $\frac 1 h = \frac12 (f_0 + f_1)'$,
%as in the proof of Thm \ref{t:harmonic_repr},
$h$ is positive if and only if  the sum $f_0 + f_1$ is a strictly increasing function.
The latter condition is weaker that the strict monotonicity of $f_0$ and $f_1$ separately, which, as seen in the previous bullet point, is itself a necessary but not sufficient condition for no arbitrage.
\end{enumerate}
\end{remark}

\noindent 
%Since the mapping between $v$ and $h$ is one-to-one,
We might also wonder whether the no butterfly-arbitrage condition \eqref{e:no_butterfly_arb} simplifies when rephrased in terms of the function $h$.
In other words: is it easier to generate arbitrage-free implied volatilities via \eqref{e:harmonic_mean} by looking for
appropriate functions $h$, than trying to look for arbitrage-free parameterizations of $v$ directly?
The condition \eqref{e:no_butterfly_arb}can be rewritten more explicitly as $v'' - \frac{v} 4 (v')^2 + \frac 1 {v} \left(1 - \frac{k \, v'}{v} \right) \ge 0$, see \cite{GathBook}.
Injecting the expression of $h$ from \eqref{e:derivative_H}, we obtain a new condition involving the three functions $h$, $h'$ and $\int_0^{\cdot} \frac 1h$ (instead of $v$, $v'$ and $v''$).
The resulting expression does not seem particularly insightful to us (concretely: it does not look more tractable than condition \eqref{e:no_butterfly_arb} itself), and is therefore not reported here.

\begin{corollary} \label{cor:h_equal_v}
Under the assumptions and notation of Theorem \ref{t:harmonic_repr},
\begin{enumerate}
\item[(i)] $v$ and $h$ coincide at the critical points of $v$:
\[
h(k) = v(k)
\Longleftrightarrow
k = 0 \mbox{ or } v'(k) = 0.
\]
\item[(ii)] The functions $v$ and $h$ satisfy the ``1/2-skew rule''
\be \label{e:one_half_skew}
v'(0) = \frac 12 h'(0).
\ee
\end{enumerate}
\end{corollary}
\begin{proof}
(i) follows immediately from \eqref{e:derivative_H}.
(ii) simply states that the $1/2$-derivative rule \eqref{e:one_half_skew} always holds at $k=0$ for a function $h$ and its harmonic mean $v$, as it can be checked by direct computation of the limit $\lim_{k \to 0} \frac{v(k) - v(0)} k$ using \eqref{e:harmonic_mean}.
\end{proof}
\medskip

\begin{remark}[An upper bound]
Since the arithmetic mean
%of a positive function
exceeds the harmonic mean, we obtain that the upper bound
\[
v(k) \le M(k) := 
\left\{
\begin{array}{ll}
\frac 1 k \int_0^k h(y) \dd y & \forall k \neq 0,
\\
v(0) & k =0.
\end{array}
\right.
\]
holds for any arbitrage-free implied volatility $v$.
\end{remark}

In Figure \ref{fig:SVI_and_h}, as an illustration we plot two examples of implied volatility smiles and their related functions $h$.
The SVI parameterization, introduced by Gatheral in 2004 \cite{SVI}, is defined by $w(k) = a + b \bigl(\rho (k - m) +\sqrt{(k-m)^2 + \sigma^2} \bigr)$, where $w(k) = v(k)^2$ denotes implied total variance.
SSVI, see \eqref{e:SSVI}, is sub-family of SVI based on three parameters $(\theta, \varphi, \rho)$ instead of five; related no-arbitrage conditions were analyzed in \cite{GathJacqSVI}.
A more detailed analysis of the SSVI framework and related applications will be carried out in section \ref{s:SSVI}.
In Figure \ref{fig:SVI_and_h}, the SVI in the left pane has the typical negative-skew shape observed in equity markets, while the positive-skew SSVI in the right pane reproduces a typical pattern observed for options on realized variance or options on the VIX index.

\begin{figure}[t]
	\centering
		\includegraphics[width=0.52\textwidth]{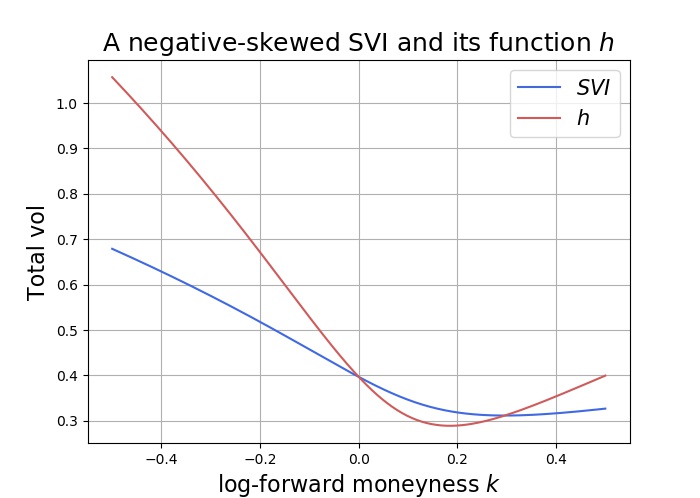}
		\hspace{-10mm}
		\includegraphics[width=0.52\textwidth]{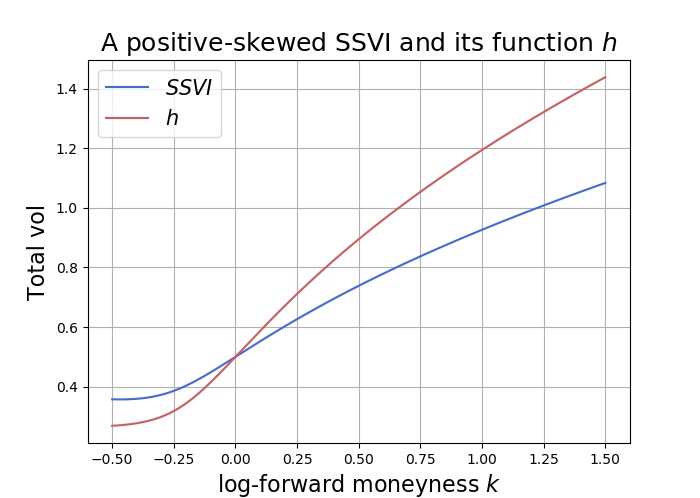}
	\caption{Examples of arbitrage-free implied volatilities and their related functions $h$ in Theorem \ref{t:harmonic_repr}.
	Left: SVI parameters are $a = 0.04, b = 0.4, \rho = -0.7, m = 0.1, \sigma = 0.2$.
	This set of parameters generate an arbitrage-free smile, as it can be checked using the procedure described in \cite{NoArbSVI}. 
	Right: SSVI parameters are $\theta=0.25, \rho = 0.7, \varphi = 3$.
	This set of parameters satisfies condition \eqref{e:noArbSSVIslice} and therefore generates an arbitrage-free smile.}
	\label{fig:SVI_and_h}
\end{figure}

\section{The normalizing transformation $f_{1/2}$}

The function $f_{1/2}(k) = \frac 12 \left(f_0(k) + f_1(k) \right)$ 
%= \frac k {v(k)} = \int_0^k \frac 1{h(y)}\dd y$ in 
we have encountered in Theorem \ref{t:harmonic_repr} is a special example of log-strike transformation: it belongs to the family of maps
%Define the linear interpolation $f_p : \R \to \R$ of the two functions $f_0$ and $f_1$ by
\[
\begin{aligned}
f_p(k)  &= p \, f_1(k) + (1-p) f_0(k) 
= \frac k{v(k)} + \Bigl(\frac 12 - p\Bigr) v(k),
\qquad \quad p \in [0,1].
\end{aligned}
\]
Owing to Theorem \ref{t:monotonicityf12}, the functions $f_p$ are strictly increasing.
When $p \neq 0$, they are also surjective:
%Denote $\mathrm{Im}(f) = f(\R)$ for a function $f$ defined on $\R$.

\begin{lemma} \label{e:image_f_p}
For every $p \in (0,1]$, $\mathrm{Im}(f_p) := f_p(\R) = \mathbb R$, so that $f_p$ is bijective from $\R$ onto $\R$.
In particular, so is $f_{1/2}$.
\end{lemma}
\begin{proof}
The no arbitrage condition $\lim_{K \to \infty} C(K) = 0$ for the call price \eqref{e:call_price} is equivalent to $\liminf_{k \to \infty} f_1(k) = \infty$, see \cite[Theorem 2.9]{Roper}.
Since $f_1(k)  = - \Bigl( \frac {|k|}{v(k)} + \frac 12 v(k) \Bigr) \le -\sqrt{2 |k|}$ for $k \mino 0$ by the arithmetic-geometric inequality, we have $\liminf_{k \to -\infty} f_1(k) = -\infty$, hence $\mathrm{Im}(f_1) = \mathbb R$.
Since $f_0(k) \maj 0$ for $k \maj 0$ and $f_0$ is decreasing,
we have $\liminf_{k \to \pm \infty} \bigl( p \, f_1(k) + (1-p) f_0(k) \bigr) = \pm \infty$, too, for every $p \in (0,1]$.
\end{proof}

\begin{remark} \label{rem:mass_zero}
Lemma \ref{e:image_f_p} is not true in general for $p=0$.
It holds that $\lim_{k \to -\infty} f_0(k) = N^{-1}(\mathbb P(X = 0))$, see \cite{TehPresentation,Fukasawa2010}, 
so that $\mathrm{Im}(f_0)$ fails to coincide with the whole $\mathbb R$ when $\mathbb P(X = 0) \neq 0$.
\end{remark}

Following Fukasawa \cite{Fukasawa2012},  we can see the map $f_{p}$ as a transformation of the log-strike variable, from $k$ to $z = f_p(k)$.
Denoting $g_p$ the inverse transformation from $\mathrm{Im}(f_p) = \R $ to $\R$,
\[
g_p(z) = f_p^{-1}(z),
\qquad z \in \R,
\]
the  so-called $p$-normalized implied volatilities $v_p$ are defined by $v_p(z) = v \left(g_p(z) \right)$, $z \in \R$.
In particular, the ``one-half normalized'' implied volatility is 
\[
v_{1/2}(z) = v \left(g_{1/2}(z) \right).
\]
The functions $v_p$ allow to  write particularly compact and elegant model-free pricing formulas for European claims, see \cite[Theorem 4.6]{Fukasawa2012} and \cite[Theorem 2.7]{DM_CM_MGFs} for a general treatment: celebrated examples are Chriss and Morokoff's formula \cite{ChrissMorok99} for the log-contract and Bergomi's formula \cite[Section 4.3.1]{Bergomi2016} for the moments of $X$ of order $p \in [0,1]$.
See Section \ref{s:pricing_formula} for more details, and for a detailed study of the SSVI case.

\begin{proposition} \label{p:harm_mean_arith_mean}
The log-strike transformation $k \to z = f_{1/2}(k)$ maps the harmonic mean representation
\[
v(k) = \frac 1{\frac 1k \int_0^k \frac 1{h(y)} \dd y}
\]
into the arithmetic mean representation
\be \label{e:arithm_mean}
v_{1/2}(z)
= \frac 1 z \int_0^z h_{1/2}(y) \dd y
\qquad \ \forall \, z \neq 0,
\ee
where $h_{1/2}(y) := h(g_{1/2}(y))$.
\end{proposition}
\begin{proof}
Let $k \in \R$. We have $k = g_{1/2}(z)$ if and only if $f_{1/2}(k) = z$.
Since $f_{1/2}(k) = \frac k {v(k)}$, we deduce $z = \frac {g_{1/2}(z)}{v_{1/2}(z)}$, or yet $z \, v_{1/2}(z) = g_{1/2}(z)$.
%\be \label{e:link_g_v_one_half}
%z \, v_{1/2}(z) = g_{1/2}(z).
%\ee
Since $g_{1/2}(0) = 0$ and
\[
g_{1/2}' (z) = \frac 1{f_{1/2}'(g_{1/2}(z))}
= h(g_{1/2}(z)) = h_{1/2}(z),
\]
we obtain $z \, v_{1/2}(z) = \int_0^z h_{1/2}(y) \dd y$, which proves \eqref{e:arithm_mean}.
\end{proof}

\textbf{Put-Call duality}. 
When $\Prob (X=0)$ (equivalently: when $\lim_{k \to -\infty} f_0(k) = -\infty$, see Remark \ref{rem:mass_zero}), the put-call symmetry relation in the Black-Scholes model implies that the mirrored function $\hat v(k) = v(-k)$ is still an arbitrage-free total implied volatility, associated to a pricing model $\hat X$ via equation \eqref{e:def_total_impl_vol}.
The distribution of the dual model $\hat X$ can be related to that of $X$ via a change of measure; see \cite{CarrLee} and section 3 in \cite{DM_CM_MGFs} for more details.
As notice in \cite{DM_CM_MGFs}, it is straightforward to check how the maps $f_p$ and the $p$-normalized implied volatilities change under the duality transformation $v \mapsto \hat v$:
\[
\hat f_p(k) = - f_{1-p}(-k),
\qquad 
%\hat g_p(k) = - g_{1-p}(-k)
\hat v_p(k) = v_{1-p}(-k),
\qquad \quad
p \in [0,1].
\]
The function $h$ is also mirrored  under the duality transformation.
Indeed, denote $\hat h$ the unique function associated to $\hat v$ in the harmonic mean representation \eqref{e:harmonic_mean}: we have
\[
\frac 1 {\hat h (k)} = \frac{\dd}{\dd x} \frac x {\hat v(x)}
=
- \frac{\dd}{\dd x} \frac {-x}{v(-x)}
=
- \frac{\dd}{\dd x} f_{1/2}(-x)
= 
f_{1/2}'(-x)
=
\frac 1 {h (-k)},
\]
so that $\hat h (k) = h(-k)$, for all $k \in \R$.

\section{Link with Dupire's local volatility} \label{s:Dupire_LV}

In this section, we consider a total implied volatility surface $v: (T,k) \in [0,\infty) \times \R \mapsto v(T,k)$, where $T$  denotes the option's time to maturity and $k = \log \frac K{F^T}$ the corresponding forward log-moneyness.
When arbitrage-free, such a function always satisfies $v(0, \cdot) \equiv 0$.
In this section, we assume that $v$ is strictly positive for strictly positive maturities and that the surface is smooth:

\begin{assumption} \label{a:assumption_v_2}
The function $(T,k) \mapsto v(T,k)$ is $C^{1,2}\left((0,\infty) \times \R\right)$ and $v(T,k) \maj 0$ for every $k \in \R$ and $T \maj 0$.
\end{assumption}

\noindent In particular, Assumption \ref{a:assumption_v} holds for every fixed $T$.
According to Theorem \ref{t:harmonic_repr}, there exists a unique strictly positive function $h$ defined over $(0,\infty) \times \R$ satisfying
\be \label{e:harmonic_mean_T}
v(T,k) = \frac 1{\frac 1k \int_0^k \frac 1{h(T,y)} \, \dd y}
\qquad \forall \, k \neq 0,  \ \forall \, T \maj 0 \,.
\ee

Assuming a market with constant interest rate $r$ and repo (or dividend) rate $q$, so that $F^T = S_0 e^{(r-q)T}$, Dupire's formula for the local volatility $\sigDup$ reads
\be \label{e:Dup_total_vol}
\begin{aligned}
\sigDup(T, k)^2
&= 
2 \, \frac
{
\frac{\dd}{\dd T} \Bigl(
e^{- r T} F^T
c_{\mathrm{BS}} \Bigl(\log \frac K {F^T}, v\bigl(T, \log \frac K {F^T} \bigr)\Bigr) 
\Bigr)
}
{
e^{- r T} F^T
K^2 \frac{\dd^2}{\dd K^2} c_{\mathrm{BS}}\Bigl(\log \frac K {F^T}, v\bigl(T, \log \frac K {F^T} \bigr)\Bigr)
}
\Biggr|_{K= F^T e^k}
\\
&= 
2 \, \frac{\partial_T v(T,k)}
{\Bigl(
v''
- \frac{v} 4 (v')^2
+ \frac 1 {v} \left(1 - \frac{k \, v'}{v}\right)^2
\Bigr)(T,k)}
\qquad \quad
\forall \, T \maj 0, \ k \in \R,
\end{aligned}
\ee
where we denote $v'(T,k):= \partial_k v(T,k)$, $v''(T,k) := \partial_{kk} v(T,k)$.
The derivation of the expression of $\sigDup$ in terms of the total implied volatility and its space-time derivatives
%follows from $\frac{\dd^2}{\dd K^2} \, c_{\mathrm{BS}}\biggl(\log \frac K {F^T}, v\Bigl(T, \log \frac K {F^T} \Bigr)\biggr)
%=  \frac 1{ F^T K} \phi(f_0(k)) 
%\biggl( v'' - \frac{v} 4 (v')^2
%+ \frac 1 {v} \left(1 - \frac{k \, v'}{v}\right)^2 
%\biggr)$ and
can be found for example in Lee \cite{Lee2005}.
It has already been noted, see again \cite{Lee2005}, that the rightmost summand in the denominator of \eqref{e:Dup_total_vol} contains precisely the squared derivative of the function $\frac k{v(T,k)}$ with respect to $k$, since $\frac 1 {v^2} \bigl(1 - \frac{k \, v'}{v}\bigr)^2 =   \bigl(\frac \dd{\dd k} \frac k{v} \bigr)^2 = \frac{1}{h^2}$, see \eqref{e:derivative_H}.
This remark provides a quick way to infer Berestycki et al.'s asymptotic formula \eqref{e:BBF_formula} from Dupire's equation \eqref{e:Dup_total_vol}.
Following \cite{Lee2005}, we can re-express \eqref{e:Dup_total_vol} in terms of the standard implied volatility $\sigBS(T, k) = \frac 1{\sqrt T} v(T,k)$: since $v' = \sqrt T \sigBS'$, $v'' = \sqrt T \sigBS''$ and $\partial_T v = \frac{\sigBS}{2 \sqrt T}  + \sqrt T \partial_T \sigBS$, we obtain 
\be \label{e:Dup_implied_vol}
\sigDup(T,k)^2 = 
\frac{\sigBS(T,k) + 2 T \partial_T \sigBS(T,k)}
{\Bigl(
T \sigBS ''
- \frac 14 T^2 \sigBS (\sigBS')^2
+ \frac 1 {\sigBS} \left(1 - \frac{k \, \sigBS'}{\sigBS}\right)^2 \Bigr)(T,k)}.
%\qquad \quad
%\forall \, T \maj 0, \ k \in \R. 
\ee
Formally taking $T=0$ inside \eqref{e:Dup_implied_vol}, one obtains an ODE for the function $\sigBS|_{T=0}$, namely
$\frac{\sigBS(0,k)^2}{\bigl(1 - \frac{k \, \sigBS'(0,k)}{\sigBS(0,k)}\bigr)^2} = \sigDup(0,k)^2$.
%\be \label{e:diff_eq_sigma_0}
%\frac{\sigBS(0,k)^2}{\left(1 - \frac{k \, \sigBS'(0,k)}{\sigBS(0,k)}\right)^2} = \sigDup(0,k)^2 \,.
%\ee
Now, it is straightforward to see that the function $\frac 1{\frac 1k \int_0^k \frac 1{\sigDup(0,y)} \dd y}$ solves this differential equation.
One could therefore conjecture that the latter	 is the limit of $\sigBS$ as $T \to 0$; Berestycki et al.'s \cite{BBF} show that this is actually the case, under the assumption that the local volatility is uniformly continuous, bounded and bounded away from zero on $[0,T] \times \R$ for some $T \maj 0$.

If we take $T$ to be fixed (but not equal to zero), dividing both sides of \eqref{e:Dup_implied_vol} by $\sigBS(T,k)$ (which is possible  under Assumption \ref{a:assumption_v_2}) and rearranging terms,  we get 
\be \label{e:identity_Sigma_Dup_2}
\begin{aligned}
\sigDup(T, k)^2
\biggl(\frac \dd{\dd k} \frac k{\sigBS(T,k)} \biggr)^2
=
\biggl[
1
+ T \biggl( \frac{2 \, \partial_T \sigBS}{\sigBS} - \frac{\sigBS''}{\sigBS} \sigDup^2  \biggr)
+ \frac 14 T^2 (\sigBS')^2 \sigDup^2
\biggr](T, k) \,.
\end{aligned}
\ee
Combining with Theorem \ref{t:harmonic_repr}, we can formulate an  equation relating the harmonic mean function $h$ to the local volatility.

\begin{proposition} \label{p:link_Dupire}
Let $v$ be a total implied volatility surface satisfying Assumption \ref{a:assumption_v_2}, and denote  $\sigBS(T,k) = \frac 1{\sqrt T} v(T,k)$ the corresponding implied volatility.
Denote $h: (0,\infty) \times \R \to (0,\infty)$ the function in  the harmonic mean representation \eqref{e:harmonic_mean_T}.
For every $(T,k)$ such that $\partial_T v(T,k) \maj 0$, we have
\be \label{e:Sigma_T}
\Sigma(T,k) :=
\frac1{\sqrt T} \, h(T,k)
= \frac{\sigDup(T, k)}
{\sqrt{ 1
+ T a(T,k)
+ T^2 b(T,k)
}} \,,
\ee
where
\[
a(T,k) = \biggl( \frac{2 \, \partial_T \sigBS}{\sigBS} - \frac{\sigBS''}{\sigBS} \sigDup^2 \biggr)(T,k),
\qquad
b(T,k) =\frac 14 \sigBS'(T,k)^2 \sigDup(T, k)^2.
\]
\end{proposition}

Since $\sigBS = \frac 1 {\sqrt T} v$, note that the function $\Sigma$ defined in \eqref{e:Sigma_T} appears in the harmonic mean representation of the standard Black-Scholes implied volatility: we have 
\be \label{e:harmonic_mean_IV}
\frac 1{\sigBS(T,k)} = \frac 1k \int_0^k \frac 1{\Sigma(T,y)} \dd y
\ee
for every $T \maj 0$ and $k \neq 0$. 

\begin{proof}
It follows from \eqref{e:Dup_total_vol} that  $\sigDup(T, k) \maj 0$ if and only if $\partial_T v(T,k) \maj 0$ for the same point $(T,k)$.
For such $(T,k)$, we can divide both sides of \eqref{e:identity_Sigma_Dup_2} by $\sigDup(T, k)^2$, obtaining that $\bigl( \frac \dd{\dd k} \frac k{\sigBS(T,k)} \bigr)^{-2}$ is equal to the square of the right hand side of \eqref{e:Sigma_T}.
According to Theorem \ref{t:harmonic_repr}, 
$\frac \dd{\dd k} \frac k{\sigBS(T,k)}$ is positive and equal to $
%\sqrt T \frac \dd{\dd k} \frac k{v(T,k)} =
\frac {\sqrt T} {h(T,k)}$, which concludes the proof of  \eqref{e:Sigma_T}.
\end{proof}

We note in passing that equation \eqref{e:identity_Sigma_Dup_2} does not allow to infer that the the function $\frac \dd{\dd k} \frac k{\sigBS(T,k)}$ is positive; this information is provided by Theorem \ref{t:harmonic_repr}.

The functions  $a$ and $b$ in Proposition \ref{p:link_Dupire} depend on the implied volatility $\sigBS$ and its space-time derivatives: in this respect, Proposition \ref{p:link_Dupire} does provide an explicit link between the local volatility $\sigDup$ and the functions $h$ and $\Sigma$ (in the sense: it does not allow to evaluate $h$ or $\Sigma$ explicitly from the knowledge of the function $\sigDup$).
Nevertheless,  we can interpret equation \eqref{e:Sigma_T} when the maturity $T$ is small: 
if the functions $a$ and $b$ remain bounded as $T \to 0$ (which is the case if $\sigBS$ and its derivatives have non-trivial limits as $T \to 0$),  equation \eqref{e:Sigma_T} takes the form
\[
\Sigma = \frac 1{\sqrt T} \, h = 
\sigDup
\bigl (1 + O(T) \bigr)\,,
\]
where the $O(T)$ correction term is precisely
$-\frac 12 T a(T,k)$.

\begin{remark}
The boundedness of the functions $a$ and $b$ as $T \to 0$ hinges on the boundedness of $\sigBS, \sigBS', \sigBS''$ and $\partial_T \sigBS$.
When the law of the underlying asset price is specified via a stochastic model, the asymptotic behavior of $\sigBS$ and of the strike derivatives $\sigBS'$ and $\sigBS''$ can be assessed for certain models including stochastic volatility models, see \cite{FukasawaSkew} and \cite{AlosLeonCurvature}, but boundedness might fail to hold in certain cases, as in rough fractional stochastic volatility models where $\lim_{T \to 0} \sigBS'(T,0) = \infty$, see \cite{Alos2007, FukasawaSkew, PricingRV}.
It is nevertheless interesting to notice that, while the function $\sigBS$ and its harmonic mean counterpart $\Sigma$ always satisfy the 1/2--skew rule 
\[
\frac {\dd \sigBS(T,k)}{\dd k} 	 \biggl|_{k=0}
= \frac 12  \frac {\dd \Sigma(T,k)}{\dd k}  \biggl|_{k=0}
\qquad \forall \, T \maj 0
\]
(see Corollary \ref{cor:h_equal_v}), Proposition \ref{p:link_Dupire} allows to characterize the possible situations where the implied volatility $\sigBS$ and the local volatility $\sigDup$ \emph{do not} satisfy the 1/2--skew rule in the short time limit\footnote{even if this criterion might be difficult to apply in practice, for it requires to evaluate the asymptotic behavior of the functions $a$ and $b$.}: they are precisely the situations where %$\lim_{T \to 0} \Sigma'(T,0) \neq \lim_{T \to 0} \sigDup'(T,0)$.
$\lim_{T \to 0} \frac \dd {\dd k} \Sigma(T,k)|_{k=0}$ $\neq  \lim_{T \to 0} \frac \dd {\dd k} \sigDup(T,k)|_{k=0}$.
\end{remark}

\subsection{Short-dated normalized implied volatility}

Since the total implied volatility $v(T,k)$ tends to zero for all $k$ as $T \to 0$, the normalizing transformations $f_{1/2}(T, k) = \frac k{v(T,k)}$ become trivial for short maturity, in the sense
\be \label{e:limit_f_one_half}
\lim_{T \to 0} f_{1/2}(T, k) =
\left\{ \begin{array}{l l}
+\infty & \mbox{if } k \maj 0
\\
0 & \mbox{if } k = 0
\\
-\infty & \mbox{if } k \mino 0 \,.
\end{array} \right.
\ee
The inverse transformations $g_{1/2}(T,\cdot) = f_{1/2}(T, \cdot)^{-1}$ therefore flatten out as $T \to 0$: we have $\lim_{T \to 0} g_{1/2}(T,z) = 0$ for every $z \in \R$.
As a consequence, the normalized implied volatility  $\sigma_{1/2}(T,z) :=
%\frac 1{\sqrt T} \, v_{1/2}(T, z) =
\frac 1{\sqrt T} \, v \left(T, g_{1/2}(T,z) \right)$ tends to a constant:
%for every $z \in \R$,
\be \label{e:limit_sigma_one_half}
\sigma_{1/2}(T, z)
=  \sigBS \left(T, g_{1/2}(T,z) \right)
\rightarrow
\sigBS \left(0,	0 \right)
\qquad 
\mbox{as } T \to 0 \,,
\ee
where the last limit holds whenever the implied volatility $\sigBS(T,k)$ tends to a limiting function $\sigBS \left(0, k\right)$
%uniformly over $k$ in a neighbourhood of zero.
uniformly in a neighbourhood of $k = 0$.

Equation \eqref{e:limit_sigma_one_half} shows that the normalized implied volatility $\sigma_{1/2}$ is a rather uninteresting object as $T$ becomes small.
We rather expect the time-rescaled function $\sigma_{1/2} \Bigl(T,  \frac z{\sqrt T}\Bigr)$ to have a non-trivial limit as $T \to 0$, see Figure \ref {fig:sigma_one_half_SSVI}.
\begin{figure}[t]
	\centering
		\includegraphics[width=0.52\textwidth]{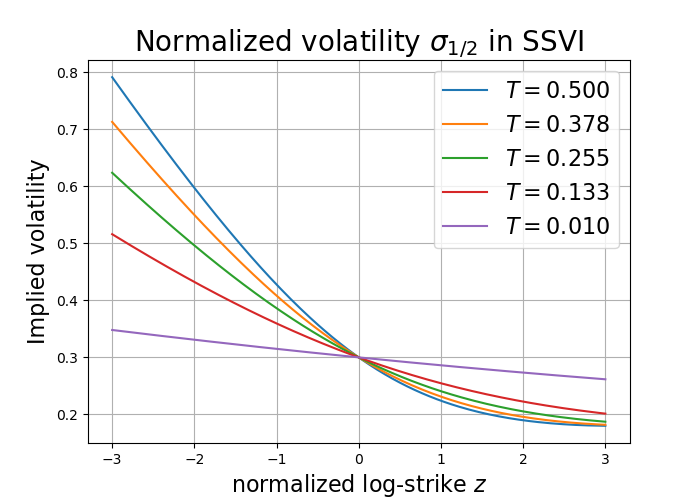}
		\hspace{-10mm}
		\includegraphics[width=0.52\textwidth]{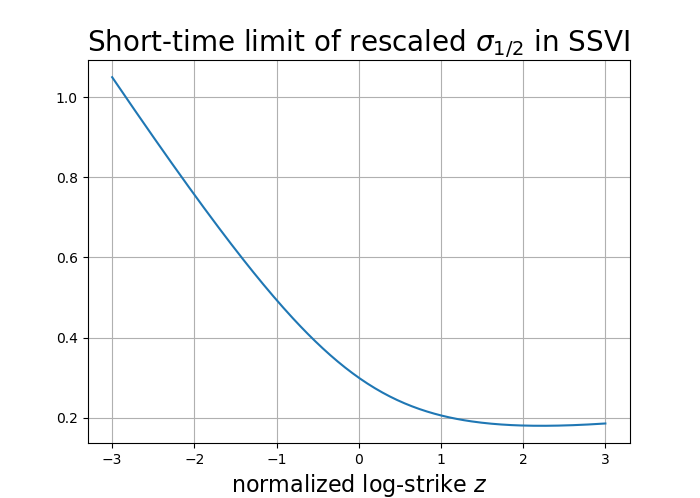}
	\caption{
	Normalized implied volatility $\sigma_{1/2}$ for the SSVI surface $v_{\mathrm{SSVI}}$ \eqref{e:SSVI} with parameters $\theta_T = \theta \, T = (0.3)^2 \, T$, $\varphi_T=\varphi=4$ and $\rho_T=\rho=-0.8$.
	The triple $(\theta, \varphi, \rho)$ satisfies the no-arbitrage condition \eqref{e:noArbSSVIslice}, so that the resulting SSVI surface
%$v_{\mathrm{SSVI}}(T,k; \theta_T, \varphi_T, \rho_T)$
is arbitrage-free for $T \le 1$.
We used the explicit expression of $\sigma_{1/2} = \frac 1{\sqrt T} v_{1/2}$ in Proposition \ref{e:SSVI_one_half_vol} to obtain the plots.
Left: the normalized implied volatility $\sigma_{1/2}$ tends to the constant value $\sqrt{\theta}$, see \eqref{e:limit_sigma_one_half}.
	Right: the time-rescaled function $\sigma_{1/2} \bigl(T,  \frac z{\sqrt T}\bigr)$ has a non-constant limit when $T \to 0$.}
	\label{fig:sigma_one_half_SSVI}
\end{figure}

We have already seen in Proposition \ref{p:harm_mean_arith_mean} that, for every $T$, the log-strike transformation $z = f_{1/2}(T, k)$ maps a coordinate system where the implied volatility $\sigBS$ is the harmonic mean of a positive function $\Sigma$ into a system where the new implied volatility $\sigma_{1/2}$ is the arithmetic mean of the transformed function $\Sigma_{1/2}$: dividing both sides of \eqref{e:arithm_mean} by $\sqrt T$ and using $\Sigma = \frac{h}{\sqrt T}$,
we have
\[
\sigma_{1/2}(T,z) = \frac 1 z \int_0^z \Sigma_{1/2}(T, y) \dd y
\qquad \forall \, T \maj 0,
\]
where $\Sigma_{1/2}(T, y) := \Sigma(T, g_{1/2}(T,y))$.
We want to investigate the consequence of this fact in the small-maturity limit, making the link with the local volatility $\sigDup$ explicit.
To this end, we also define the normalized local volatility
\[
\sigma_{\mathrm{Dup}, 1/2} \bigl(T,  z \bigr) := \sigDup \bigl(T, g_{1/2} (T,  z ) \bigr)
\quad
\forall \, T \maj 0, \ z \in \R.
\]

Theorem \ref {t:arithm_mean_formula} below essentially states that, if the implied volatility surface $\sigBS$ has a non-degenerate behavior for small maturity (in the precise sense of Assumption \ref{a:short_time_limit}), then the rescaled normalized implied volatility $\sigma_{1/2} \bigl(T,  \frac z{\sqrt T}\bigr)$ and local volatility $\sigma_{\mathrm{Dup}, 1/2}  \bigl(T,  \frac z{\sqrt T}\bigr)$ also have non-trivial short-maturity limits, and these limits obey an \emph{arithmetic} mean formula.

\begin{assumption} \label{a:short_time_limit}
The implied volatility $\sigBS$ and the  local volatility $\sigDup$ have non-trivial limits in the short maturity regime, in the sense:
\begin{itemize}
\item[(i)] There exists a strictly positive and differentiable function $\sigBS(0,\cdot)$  such  that $\sigBS(T,k) \to \sigBS(0,k)$ together with $\partial_k \sigBS(T,k) \to \partial_k \sigBS(0,k)$ as $T \to 0$, uniformly over $k$ in compact sets.
\item[(ii)] The function $k \mapsto \frac k{\sigBS(0,k)}$ is strictly increasing from $\R$ onto $\R$.
\item[(iii)] The local volatility has a short-maturity limit: $\sigDup(T,k) \to \sigDup(0,k)$ as $T \to 0$ for some strictly positive function $\sigDup(0,\cdot)$, uniformly over $k$ in compact sets.
\end{itemize}
\end{assumption}

\begin{theorem}[Arithmetic mean formula for the short-dated implied volatility $\sigma_{1/2}$] \label{t:arithm_mean_formula}
Let Assumption \ref{a:short_time_limit} be in force.
Then, the rescaled normalized implied volatility $\sigma_{1/2} \bigl(T,  \frac z{\sqrt T}\bigr)$ and normalized local volatility $\sigma_{\mathrm{Dup}, 1/2} \bigl(T,  \frac z{\sqrt T}\bigr)$ have non-trivial limits as $T \to 0$:
\be \label{e:limits_normalized_implied_local_vol}
\sigma_{1/2} \biggl(T,  \frac z{\sqrt T}\biggr)
\longrightarrow
\sigma_{1/2}(z),
\qquad
\sigma_{\mathrm{Dup}, 1/2} \biggl(T,  \frac z{\sqrt T}\biggr)
\longrightarrow
\sigma_{\mathrm{Dup}, 1/2} (z),
\ee
where the convergence is uniform over $z$ in compact sets.
If, in addition, the two functions $\partial_T \sigBS(T,k)$ and $\partial_{kk} \sigBS(T,k)$ remain bounded as $T \to 0$, we have the following arithmetic mean formula
%for the small time limit of the normalized implied volatility:
\be \label{e:short_time_arith_mean_formula}
\sigma_{1/2}(z)
= \frac 1z \int_0^z 
\sigma_{\mathrm{Dup}, 1/2} (y) \, \dd y \,,
\qquad
\forall \ z \in \R.
\ee
\end{theorem}
\begin{proof}
Denote $\overg_{1/2}$ the inverse of the function $k \mapsto \frac k {\sigBS(0,k)}$.
We claim that
%the rescaled transformations 
\be \label{e:rescaled_g_limit}
g_{1/2}\biggl(T,  \frac z{\sqrt T}\biggr) \longrightarrow
\overg_{1/2}(z)
\qquad
\mbox{as $T \to 0$, uniformly over $z$ in compact sets}.
\ee
Let us  postpone the proof of \eqref{e:rescaled_g_limit}to Appendix \ref{s:app} and focus on the asymptotic behavior of $\Sigma_{1/2} \Bigl(T,  \frac z{\sqrt T}\Bigr)$ and $\sigma_{1/2} \Bigl(T,  \frac z{\sqrt T}\Bigr)$.
Recall from \eqref{e:harmonic_mean_IV} that $\frac 1 {\Sigma(T,k)} = \frac {\dd}{\dd k} \frac k {\sigBS(T,k)} = \frac 1 {\sigBS(T,k)} \Bigl( 1 - k \, \frac{\sigBS'(T,k)}{\sigBS(T,k)} \Bigr)$ for every $T \maj 0$.
According to point (ii) in Assumption \ref{a:short_time_limit}, we can define the function $\Sigma(0,k)$ from $\frac 1 {\Sigma(0,k)} = \frac {\dd}{\dd k} \frac k {\sigBS(0,k)}$, so that
\be \label{e:uniform_limit_Sigma}
\Sigma(T,k) 
= 
\frac{\sigBS(T,k)}
{\Bigl( 1 - k \, \frac{\sigBS'(T,k)}{\sigBS(T,k)} \Bigr)}
\longrightarrow
\frac{\sigBS(0,k)}
{\Bigl( 1 - k \, \frac{\sigBS'(0,k)}{\sigBS(0,k)} \Bigr)}
= 
\Sigma(0,k)
\qquad
\mbox{as } T \to 0
\ee
holds uniformly over $k$ in compact sets, owing to Assumption \ref{a:short_time_limit} (i).
For simplicity, denote $k^z_T = g_{1/2}\bigl(T,  \frac z{\sqrt T}\bigr)$ and $k^z_0 = \overg_{1/2}(z)$.
For every compact set $C \subset \R$, 
\[
\sup_{z \in C} \left|\Sigma(T, k^z_T) - \Sigma(0, k^z_0) \right|
\le
\sup_{z \in C} \left|\Sigma(T, k^z_T) - \Sigma(0, k^z_T) \right|
+
\sup_{z \in C} \left|\Sigma(0, k^z_T) - \Sigma(0, k^z_0) \right|.
\]
The first term on the right hand side tends to zero as $T \to 0$ because, on the one hand, $\Sigma(T, k) - \Sigma(0, k) \to 0$ uniformly over $k$ in compacts,
%\eqref{e:uniform_limit_Sigma}
and on the other hand, according to \eqref{e:rescaled_g_limit}, $\sup_{z \in C} |k^z_T| \le \sup_{z \in C} |\overg_{1/2}(z)| + 1$ for $T$ small enough.
The second term also tends to zero as $T \to 0$ because $k^z_T \to k^z_0$ uniformly over $z \in C$ and $\Sigma(0, \cdot)$ is continuous by assumption.
Summing up, we have shown that
\be \label{e:convergence_Sigma}
\Sigma \biggl(T,  g_{1/2} \biggl( T,  \frac z{\sqrt T}\biggr) \biggr)
=
\Sigma_{1/2} \biggl(T,  \frac z{\sqrt T}\biggr)
\longrightarrow
\Sigma \bigl(0,  \overg_{1/2}(z) \bigr)
=: \Sigma_{1/2} (z)
\quad
\mbox{as }T \to 0,
\ee
uniformly over $z$ in compact sets.
Concerning $\sigma_{1/2}$, it is immediate to see that \eqref{e:convergence_Sigma} implies that the limit 
\be \label{e:limit_sigma_one_half_rescaled}
\sigma_{1/2} \biggl(T,  \frac z{\sqrt T}\biggr)
= 
\frac {\sqrt T} z \int_0^{z / \sqrt T} \Sigma_{1/2}(T, y) \dd y 
= 
\frac 1 z \int_0^{z} \Sigma_{1/2}\biggl(T,  \frac x{\sqrt T}\biggr) \dd x
\longrightarrow
\frac 1 z \int_0^{z} \Sigma_{1/2}(x) \dd x
\ee
holds as $T \to 0$ uniformly over $z$ in compact sets, too, which proves the first part of \eqref{e:limits_normalized_implied_local_vol} with $\sigma_{1/2}(z) = \frac 1 z \int_0^{z} \Sigma_{1/2}(x) \dd x$.

Let us move to the normalized local volatility.
It follows from \eqref{e:rescaled_g_limit} and Assumption \ref{a:short_time_limit} (iii) that 
$\sigma_{\mathrm{Dup}, 1/2} \Bigl(T,  \frac z{\sqrt T} \Bigr) = 
\sigDup \Bigl(T,  g_{1/2} \Bigl( T,  \frac z{\sqrt T}\Bigr) \Bigr) \to \sigDup \bigl(0,  \overg_{1/2}(z) \bigr)$ as $T \to 0$, hence the second limit in \eqref{e:limits_normalized_implied_local_vol} follows with $\sigma_{\mathrm{Dup}, 1/2} (z) := \sigDup \bigl(0,  \overg_{1/2}(z) \bigr)$.
Using the additional assumption on the boundedness of the functions $\partial_T \sigBS$ and $\partial_{kk} \sigBS$ as $T \to 0$, we can see that the functions $a$ and $b$ defined in Proposition \ref{p:link_Dupire} remain bounded as $T \to 0$, which implies that $1 + T a\bigl(T,  g_{1/2} \bigl( T,  \frac z{\sqrt T}\bigr) \bigr) + T^2 b\bigl(T,  g_{1/2} \bigl( T,  \frac z{\sqrt T}\bigr) \bigr) \to 1$ as $T \to 0$.
Therefore, it follows from equation \eqref{e:Sigma_T} that
\[
\begin{aligned}
\Sigma \biggl(T,  g_{1/2} \biggl(T,  \frac z{\sqrt T}\biggr) \biggr)
\sim 
\sigDup \biggl(T,  g_{1/2} \biggl( T,  \frac z{\sqrt T}\biggr) \biggr)
%&= \sigma_{\mathrm{Dup}, 1/2} \Bigl( T,  \frac z{\sqrt T}\Bigr)
\longrightarrow 
%\sigDup \bigl(0,  \overg_{1/2}(z) \bigr)
\sigma_{\mathrm{Dup}, 1/2} (z)
\quad \mbox{as } T \to 0.
\end{aligned}
\]
Comparing with \eqref{e:convergence_Sigma},  we identify $\Sigma_{1/2} (z)$ with $\sigma_{\mathrm{Dup}, 1/2} (z)$, hence formula \eqref{e:short_time_arith_mean_formula} follows from \eqref{e:limit_sigma_one_half_rescaled}.
\end{proof}

\section{A parametric framework: SSVI} \label{s:SSVI}

The SSVI parameterisation for total implied variance $w = v^2$ proposed by Gatheral and Jacquier \cite{GathJacqSVI} reads $w_{\mathrm{SSVI}}(k) = \frac{\theta}2 \left(1 + \rho \varphi  k + \sqrt{(\varphi k + \rho)^2 + 1 -\rho^2} \right)$.
The corresponding total implied volatility is
\begin{equation} \label{e:SSVI}
v_{\mathrm{SSVI}}(k) = \sqrt{ \frac{\theta}2 \left(1 + \rho \varphi  k + \sqrt{\Delta(k)} \right)},
\qquad
\Delta(k) = (\varphi k + \rho)^2 + 1 -\rho^2.
\end{equation}
Equation \eqref{e:SSVI} parametrizes a slice of the implied volatility surface at fixed maturity: the SSVI parameters are $\theta \maj 0$, $\varphi \maj 0$ and $\rho \in (-1,1)$.
Since $\min_{\{k \in \R\}} w_{\mathrm{SSVI}}(k) =  \frac \theta 2 \bigl( 1 - 2 \rho^2 + \sqrt{3 \rho^2 +1} \bigr) \maj 0$ for every $\rho \in (-1,1)$, $v_{\mathrm{SSVI}}$ satisfies Assumption \ref{a:assumption_v} for any triple $(\theta, \varphi, \rho)$ as above.
Theorem 4.2 in \cite{GathJacqSVI} proves that $v_{\mathrm{SSVI}}$ is a total implied volatility free of arbitrage (for the given maturity) if the  following sufficient conditions are satisfied:
\be \label{e:noArbSSVIslice}
\theta \varphi (1+|\rho|) < 4;
\qquad 
\theta \varphi^2 (1+|\rho|) \le 4.
\ee
The condition $\theta \varphi (1+|\rho|) \le 4$ (including the equality) is known to be necessary, see \cite[Lemma 4.2]{GathJacqSVI}, while the second inequality 	in \eqref{e:noArbSSVIslice} is not necessary in general.

We can evaluate the function $h$ for SSVI from
\eqref{e:derivative_H}:
\[
h_{\mathrm{SSVI}}(k)
= \left(\frac 1{v(k)}- \frac k {2 \, v^3(k)} w' \right)^{-1}
= \left(\frac 1{v(k)}- \frac {k \theta} {4 v^3(k)}
\biggl( \rho \varphi + \frac{\varphi k + \rho}{\sqrt{\Delta(k)}} \biggr)
\right)^{-1}
\]
where $v = v_\mathrm{SSVI}$ and $w=w_\mathrm{SSVI}$.
Moreover it turns out that, for the SSVI parameterisation, it is possible to explicitly compute the inverse transformation $g_{1/2}$ and the normalized implied volatility $v_{1/2}$.

\begin{proposition} \label{e:SSVI_one_half_vol}
Assume that the triple $(\theta, \varphi, \rho)$ satisfies \eqref{e:noArbSSVIslice}.
Then, the inverse transformation $g_{1/2} = f_{1/2}^{-1}$ for the SSVI parameterisation \eqref{e:SSVI} is given by
\[
g_{1/2}(z) =
\frac 12 \bigl(
\theta\rho \varphi z^2
+  z \sqrt{\theta^2 \varphi^2 z^2 + 4 \theta}
\bigr),
\qquad z \in \R;
\]
and the $1/2$--normalized SSVI implied volatility is
\be \label{e:v_one_half_SSVI}
v_{1/2}(z) =
\frac 12 \bigl(
\theta \rho  \varphi z
+ \sqrt{\theta^2 \varphi^2 z^2 + 4 \theta}
\bigr),
\qquad z \in \R.
\ee
\end{proposition}

\begin{remark}
The normalized SSVI volatility $v_{1/2}$ is asymptotically linear for large arguments $z$, with $\lim_{z \pm \infty} \frac{v_{1/2}(z)}{|z|} = \theta \varphi (1 \pm \rho)$, as opposed to the SSVI implied volatility \eqref{e:SSVI}, which is
%asymptotic to $\mathrm{const}. \times \sqrt{|k|}$
proportional to $\sqrt{|k|}$ for  large values of $|k|$.   
\end{remark}

\begin{proof}[Proof of Proposition \ref{e:SSVI_one_half_vol}]
Denoting $v(k) = v_\mathrm{SSVI}$ for simplicity, we have
\be \label{e:starting_point_g}
f_{1/2}(k)
%= \frac k {v(k)}
= z 
\Longleftrightarrow 
k = v(k) z. 
%\mbox{ and } \mathrm{sign}(k) = \mathrm{sign}(z).
\ee
By the invertibility of $f_{1/2}$ and Lemma \ref{e:image_f_p}, we know that the equation on the right hand side of \eqref{e:starting_point_g} has a unique solution $k$ for every $z \in \R$, which coincides with $g_{1/2}(z)$.
We already know that $g_{1/2}(z) = 0$ for $z=0$, hence we assume $z \neq 0$ in what follows, which implies $k \neq 0$.

By squaring both sides in \eqref{e:starting_point_g}, we obtain that, for every $z$, the equation $k^2 = v(k)^2 z^2$ has exactly two solutions $k_\pm$ (given by $g_{1/2}(z)$ and $g_{1/2}(-z)$).
Using the SSVI formula \eqref{e:SSVI}, the equation $k^2 = v(k)^2 z^2$ is equivalent to $\sqrt{\Delta(k)} = \frac{2 k^2}{\theta z^2} - 1 - \rho \varphi k$.
If we square both sides again, passing to the quartic equation $\Delta(k) = \Bigl(\frac{2 k^2}{\theta z^2} - 1 - \rho \varphi k\Bigr)^2$, we might add spurious solutions:
it turns out that this quartic equation has only two roots, which we can therefore identify with $k_\pm$.
Let us work this out: expanding the squares in $\Delta(k) = \Bigl(\frac{2 k^2}{\theta z^2} - 1 - \rho \varphi k\Bigr)^2$, after some cancellations and rearrangements we obtain
\[
\frac{4}{\theta^2 z^4} k^4
- \frac{4 \varphi \rho}{\theta z^2} k^3
+ \biggl(\rho \varphi^2 - \varphi^2 - \frac{4}{\theta z^2} \biggr) k^2 = 0.
\]
The special feature of the equation above it that it has no constant term and no linear term in $k$. Dividing by $k^2$, we are left with the quadratic equation $\frac{4}{\theta^2 z^4} k^2 - \frac{4 \varphi \rho}{\theta z^2} k + c = 0$,
%for the unknown $k$,
where $c = \rho \varphi^2 - \varphi^2 - \frac{4}{\theta z^2}$.
The two roots 
\[
k_{\pm}
= \frac{\theta z^2}{2}
\biggl(
\rho \varphi  \pm 
\sqrt{ \frac 4 {\theta z^2} + \varphi^2} 
\biggr)
\]
are therefore the two solutions of $k^2 = v(k)^2 z^2$.
Since $\sqrt{\frac 4 {\theta z^2} + \varphi^2} \maj \varphi \ge |\rho \varphi|$, we have $k_+  \maj 0$ and $k_- \mino 0$ for every $z$, which allows us to identify $g_{1/2}(z)$ with $k_+$ for $z \maj 0$, resp.\til $k_-$ for $z \mino 0$.
Overall, we have obtained
\[
g_{1/2}(z)
= \frac{\theta z^2}{2}
\biggl(
\rho \varphi  + \mathrm{sign}(z)
\sqrt{ \frac 4 {\theta z^2} + \varphi^2} 
\biggr)
= 
\frac 12 
\biggl(
\theta \rho \varphi z^2 
+ z \sqrt{\theta^2 \varphi^2 z^2 + 4 \theta} 
\biggr),
\quad \forall \, z \in \R.
\]
Finally, since \eqref{e:starting_point_g} implies $g_{1/2}(z) = v_{1/2}(z) z$, we have $v_{1/2}(z) = \frac{g_{1/2}(z) }z$ for every $z \neq 0$, from which \eqref{e:v_one_half_SSVI} follows.
\end{proof}

\subsection{A pricing formula for $\mathbb{E}\bigl[\sqrt{X} \, \bigr]$}
\label{s:pricing_formula}

When $v$ is arbitrage-free and twice differentiable (as in the case of the SSVI parameterisation \eqref{e:SSVI}), the law of the random variable $X$ in \eqref{e:def_total_impl_vol} has a absolutely continuous part with density $f_X(K) = \frac{\dd^2 C(K)}{\dd K^2} 1_{K \maj 0}$ (the function $C$ being defined in \eqref{e:call_price}),
%c_{\mathrm{BS}}\Bigl(\log \frac K {F}, v\bigl(\log \frac K {F} \bigr)\Bigr) 1_{K \maj 0}$,
plus a possible atom at zero with mass $\Prob(X=0) = 1+\lim_{K \to 0} \partial_K C(K)$.
According to Remark \ref{rem:mass_zero}, we have $\Prob(X=0) = 0$ if and only if $\lim_{k \to -\infty} f_0(k) = -\infty$: a sufficient condition for this to hold is $\beta_- := \limsup_{k \to -\infty} \frac{v(k)^2}{|k|} \mino 2$, see e.g.\til\cite{Roper}.
Under such condition, any integrable claim $h(X)$ can be priced as
\be \label{e:pricing_h}
\esp[h(X)] = \int_0^\infty	h(K)  %f_X(K) 
\frac{\dd^2 C(K)}{\dd K^2} \dd K
= \int_0^\infty	h(K) %f_X(K) 
\frac{\dd^2}{\dd K^2} c_{\mathrm{BS}}(k,v(k))|_{k = \log \frac K {F}} \dd K \,.
\ee
For some payoff functions $h$, it is possible to convert equation \eqref{e:pricing_h} into a particularly compact and elegant formula written in terms of one of the normalized implied volatilities $v_p$, see \cite[Theorem 4.6]{Fukasawa2012} and \cite[Theorem 2.7]{DM_CM_MGFs}.\footnote{The derivation of such a formula goes through the following steps: express $\frac{\dd^2}{\dd K^2} c_{\mathrm{BS}}(k,v(k))|_{k = \log \frac K {F}}$ in terms of Black-Scholes greeks and the derivatives of $v$ up to order two,  integrate by parts (carefully checking that the boundary terms vanish), and finally apply the log-strike transformation $z = f_p(k)$.}
Celebrated examples are Chriss and Morokoff's formula \cite{ChrissMorok99} for the log-contract,
$\esp[-2 \log(X)] = \int_{\R} v_0(z)^2 \phi(z) \dd z$ (recall that $\phi$ denotes the standard normal density), and Bergomi's formula \cite[Section 4.3.1]{Bergomi2016} for the moments of $X$ of order $p \in [0,1]$,
$\mathbb{E}\left[ X^p \right] = \int_{	\mathbb R} e^{\frac12 p(p-1) v_{p}(z)^2} \phi(z) \dd z$.
%\be \label{e:integral_formula_X_p}
%\mathbb{E}\left[ X^p \right]
%= \int_{	\mathbb R} e^{\frac12 p(p-1) v_{p}(z)^2} \phi(z) dz.
%\ee
In particular, we have
\be \label{e:integral_formula_X_one_half}
\mathbb{E}\left[ \sqrt X \right]
= \int_{	\mathbb R} e^{- \frac 18 v_{1/2}(z)^2} \phi(z) \dd z.
\ee
Formula \eqref{e:integral_formula_X_one_half} is appealing because it only requires to know the values of $v_{1/2}$ (on, say, a set of quadrature points) and does not require to access the values of the derivatives of $v_{1/2}$.
When the number of liquid log-moneyness $k_i$ and implied volatilities $v^{\mathrm{mkt}}(k_i)$ observed on the market is sufficiently large, one can discretize the right hand side of \eqref{e:integral_formula_X_one_half} on the points $z_i = f_{1/2}(k_i) = \frac{k_i}{v^{\mathrm{mkt}}(k_i)}$, thus obtaining a model-free pricing formula for the claim $\sqrt X$.
On the other hand, when the starting point of pricing operations is to fit market option data with a volatility parameterization (such as SSVI), possibly because market data is relatively scarce,
it is of course interesting to be provided with explicit pricing formulas, instead of having to rely on  numerical integration of  \eqref{e:integral_formula_X_one_half}.
When $v$ is given by $v_{\mathrm{SSVI}}$, this is the program we carry out in the rest of this section.
%derive an explicit expression for \eqref{e:integral_formula_X_one_half}.

When the underlying is the annualized realized variance of an observable asset $S$,
\[
X_T = \frac 1 T \sum_{0 \mino t_i \le T} \left(\log(S_{t_i}) - \log(S_{t_{i-1}})\right)^2,
\]
or its continuous counterpart $X_T = \frac 1 T \langle \log S \rangle_T$ when $S$ is a semi-martingale model, then $\esp [\sqrt X_T]$ represents the (undiscounted) price of realized volatility over the interval $[0,T]$, aka the fair strike of the volatility swap (or yet again, the vol swap forward volatility).
If SSVI is used to fit options on realized variance, Proposition \ref{e:square_root_X_explicit formula} below provides an explicit formula for the volatility swap.

First, note that $\beta_-(\mathrm{SSVI}) = \lim_{k \to -\infty} \frac{v_{\mathrm{SSVI}}(k)^2}{|k|} = \frac 12 \theta \varphi (1 - \rho)$, hence $\beta_-(\mathrm{SSVI}) \mino 2$ is granted by  the first condition in \eqref{e:noArbSSVIslice}.
%(a more extensive discussion of the extreme case $\beta_-= 2$ can be found in \cite{DM_CM_MGFs}).
Since
\[
-\frac 18 v_{1/2}(z)^2 - \frac 12 z^2
= 
- \Bigl( \frac 1{32} \theta^2 \varphi^2 (1+ \rho^2) + \frac 12 \Bigr) z^2
- \frac 1{8} \theta \rho \varphi z \sqrt{\frac{\theta^2 \varphi^2} 4 z^2 + \theta} 
- \frac \theta 8 \,,
\]
from \eqref{e:integral_formula_X_one_half} we have
\be \label{e:formula_sqrt_1}
\esp \left[\sqrt{X}  \right]
=
\frac{2 \, e^{-\frac \theta 8} } { \theta \varphi}
\int_{\R}
e^{
- A_{\theta} y^2
- \frac 14 \rho \, y  \sqrt{y^2 + \theta}
}
\frac{\dd y}{\sqrt{2 \pi}}
\qquad
\mbox{with } A_{\theta} = \frac 18 (1 +\rho^2) + \frac 2{\theta^2 \varphi^2} \,,
\ee
where we have applied the change of variable $y = \frac{\theta \varphi} 2 z$ to the right hand side of \eqref{e:integral_formula_X_one_half}.

Denote 
\[
I(\theta, \varphi, \rho) := \int_{\R}
e^{
- A_{\theta} y^2
- \frac 14 \rho \, y  \sqrt{y^2 + \theta}
}
\frac{\dd y}{\sqrt{2 \pi}}
\]
%the integral on the right hand side of \eqref{e:formula_sqrt_1}.
When $\rho = 0$, we simply have to evaluate a Gaussian integral $I(\theta, \varphi, 0) =  \int_{\R} e^{- A_{\theta} y^2} \frac{\dd y}{\sqrt{2 \pi}}
= \frac 1{\sqrt{2 A_\theta}}$.
When $\rho$ is different from zero, we did not manage to find an explicit expression for $I(\theta, \varphi, \rho)$ (neither did WolframAlpha online integrator\footnote{\url{www.wolframalpha.com}}),
but we can investigate the asymptotic behavior of the integral in limiting parameter regimes.
When SSVI \eqref{e:SSVI} is calibrated to market option data, typical values of $\varphi$ are of order $10^0$ or $10^1$,  see \cite{eSSVI}, while $\rho$ can approach $-1$ (for equity index smiles) or $1$ (as in the case of realized variance options, which usually display a positive implied volatility skew, see \cite{Drimus2012}, just as options on the VIX index).
Being an ATM implied variance, the parameter $\theta$ is usually small, often of order $10^{-2}$:
it seems reasonable, then, to look for an asymptotic approximation of $I(\theta, \varphi, \rho)$ as $\theta \to 0$.
A quick inspection reveals that the asymptotic behavior of the integral
in this regime is governed by the factor $e^{- A_{\theta} y^2}$:
since $A_\theta \to \infty$ as $\theta \to 0$, the function $e^{- A_{\theta} y^2}$ tends to zero exponentially fast for every $y \neq 0$ and, in the spirit of Laplace's method for integral approximation, the asymptotic behavior %of the integral
is influenced by the second factor $f_\theta(y)  = e^{- \frac 14 \rho \, y  \sqrt{y^2 + \theta}}$ only in a neighbourhood of $y = 0$.
We have to pay attention to the fact that the function $f_\theta(\cdot)$ also depends on $\theta$, but this does not add substantial difficulties to the analysis.

\begin{proposition}[Volatility swap value when $X_T$ models realized variance] \label{e:square_root_X_explicit formula}
For every $\rho \in (-1,1)$ and $\varphi \maj 0$, the no-arbitrage condition \eqref{e:noArbSSVIslice} is satisfied for $\theta$ small enough.
When the total implied volatility of the underlying $X_T$ is given by SSVI \eqref{e:SSVI}, the following asymptotic formula holds:
\be \label{e:asymptotics_expectation}
\esp \left[\sqrt{X_T}  \right]
=
 \frac{e^{-\frac \theta 8}}{\sqrt{ 1 + \frac{1+\rho^2}{16} \theta^2 \varphi^2}}(1 + o(1))
\qquad
\mbox{as } \theta \to 0.
\ee
Moreover, if $\rho = 0$, the above formula is exact: we have $\esp \bigl[\sqrt{X_T} \, \bigr] = \frac{e^{-\frac \theta 8}}{\sqrt{ 1 + \frac{\theta^2 \varphi^2}{16}}}$ for every couple $(\theta, \varphi)$ satisfying condition \eqref{e:noArbSSVIslice}.
\end{proposition}

The proof of Proposition \ref{e:square_root_X_explicit formula} is postponed to the Appendix.
We test the asymptotic formula \eqref{e:asymptotics_expectation} numerically under two parameter configurations, one with positive and one with negative ATM skew; the results are shown in Figures \ref{fig:formula_SSVI_pos_slope} and  \ref{fig:formula_SSVI_neg_slope}.
Formula \eqref{e:asymptotics_expectation} appears to be very accurate even for values of $\theta$ corresponding to ATM implied volatilities of $100\%$ as in Figure \ref{fig:formula_SSVI_pos_slope}
(in the right pane, relative errors remain below $2\%$ for $\theta=1$).
Figure \ref{fig:formula_SSVI_neg_slope} displays a situation that is more typical of equity indices.
Overall, we deem that the accuracy of the asymptotic approximation in Proposition \ref{e:square_root_X_explicit formula} is more than satisfactory in a wide range of market conditions,
including realized variance options or possibly stressed equity markets.

\begin{figure}[t]
	\centering
		\includegraphics[width=0.52\textwidth]{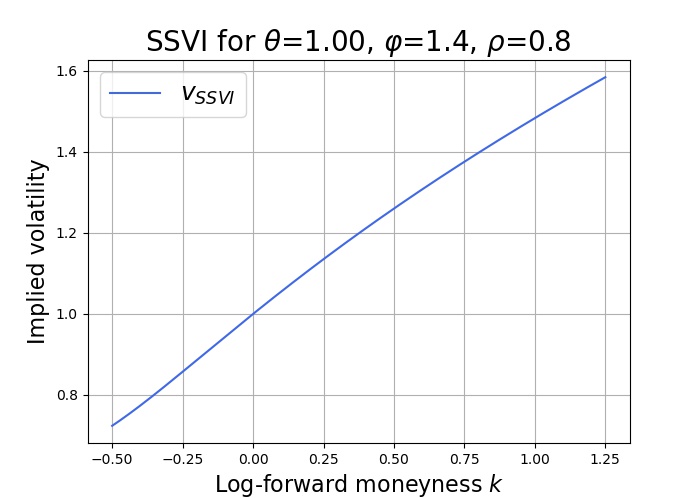}
		\hspace{-10mm}
		\includegraphics[width=0.52\textwidth]{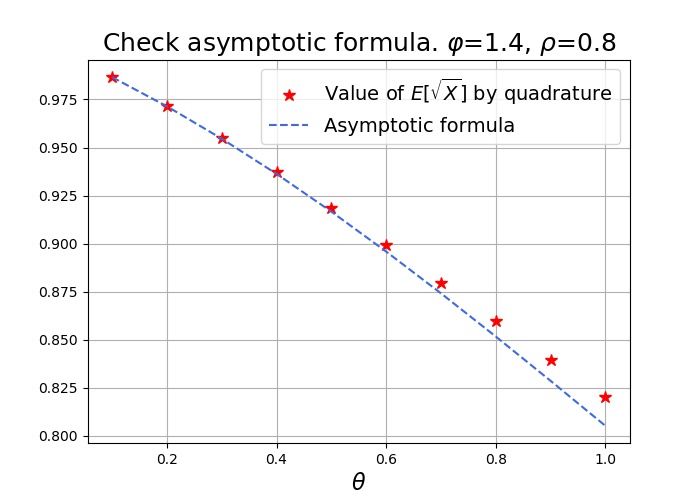}
	\caption{Left: an example of SSVI parameterisation with positive skew, typical of options on realized variance. Right: comparison of the asymptotic formula \eqref{e:asymptotics_expectation} as $\theta \to 0$ (dashed line) with the values of $\esp \bigl[\sqrt{X} \, \bigr]$ obtained by quadrature of the right hand side of \eqref{e:formula_sqrt_1}, for different values of $\theta$.}
	\label{fig:formula_SSVI_pos_slope}
\end{figure}

\begin{figure}[t]
	\centering
		\includegraphics[width=0.52\textwidth]{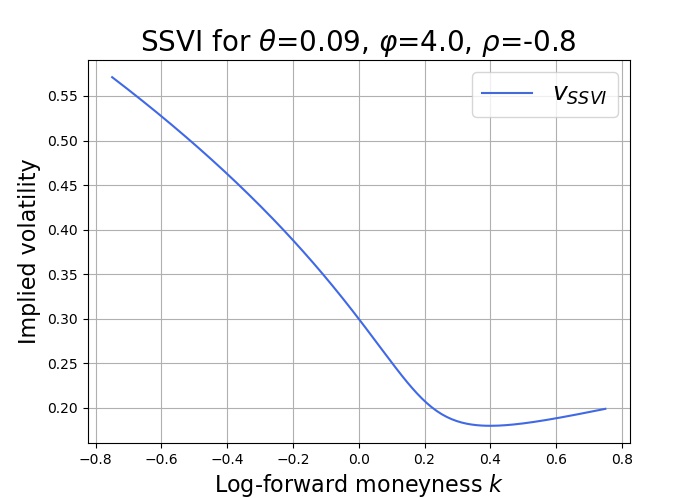}
		\hspace{-10mm}
		\includegraphics[width=0.52\textwidth]{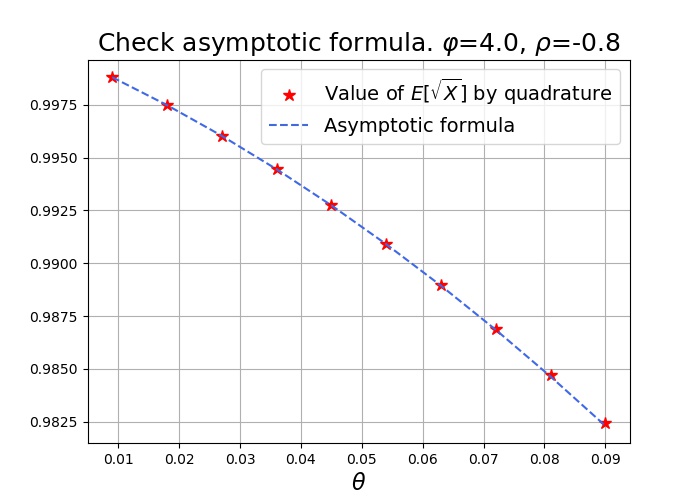}
	\caption{Same functions as in Figure \ref{fig:formula_SSVI_pos_slope} (left: SSVI implied volatility; right: values of $\esp \bigl[\sqrt{X} \bigr]$ for different values of $\theta$), now for a SSVI parameterisation with negative ATM skew.}
\label{fig:formula_SSVI_neg_slope}
\end{figure}

\appendix

\section{Appendix} \label{s:app}

\begin{proof}[Proof of \eqref{e:rescaled_g_limit}]
For every $T \maj 0$, define the time--rescaled transformations
\[
\begin{aligned}
&\overline f_{1/2}(T, k)
:= \sqrt T  f_{1/2}(T, k)
=  \frac k{\sigBS(T,k)}
\\
&\overline g_{1/2}(T, z)
:= \overline f_{1/2}(T, \cdot)^{-1}(z)
= g_{1/2}\Bigl(T,  \frac z{\sqrt T}\Bigr) \,.
\end{aligned}
\]
%According to Theorem \ref{t:monotonicityf12}, $\overline f_{1/2}(T,\cdot)$ and $\overg_{1/2}(T,\cdot)$ are strictly increasing from $\R$ onto $\R$, for every $T \maj 0$.
Since the functions $\overg_{1/2}(T,\cdot)$ are strictly increasing from $\R$ onto $\R$ for every $T$, it is sufficient to show that the convergence $\lim_{T \to 0} \overline g_{1/2}(T, z) = \overline g_{1/2}(z)$ holds pointwise. 
Thanks to Assumption \eqref{a:short_time_limit} (i), we know that $\overline f_{1/2}(T, k)$ tends to $\overline f_{1/2}(k) := \frac k{\sigBS(0,k)}$ uniformly over $k$ in compact sets.
We have to transfer the uniform convergence of $\overline f_{1/2}(T, \cdot)$ towards $\overline f_{1/2}(k) := \frac k{\sigBS(0,k)}$, granted by Assumption \eqref{a:short_time_limit} (i), to pointwise convergence of the inverse functions $\overg_{1/2}(T, \cdot)$.
This is a rather standard procedure: fix 
$z \in \R$, and denote for simplicity $k^z_T = \overg_{1/2}(T,  z)$.
First, it is not difficult to see that $k^z_T$ remains in a compact set as $T \to 0$.
Then, since
\[
\overline f_{1/2}(k^z_T) =
%\overline f_{1/2}(T, k^z_T) 
z
+ 
\bigl(
\overline f_{1/2}(k^z_T) - \overline f_{1/2}(T,k^z_T)
\bigr)
\longrightarrow z
\qquad \mbox{as } T \to 0
\]
due to the uniform convergence on compact sets of $\overline f_{1/2}(T,\cdot)$ to $\overline f_{1/2}(\cdot)$, 
using the continuity of $\overline f_{1/2}^{-1}(\cdot)$ we obtain $k^z_T
%= \overg_{1/2}(T,z) 
\to f_{1/2}^{-1}(z) = \overg_{1/2}(z)$ as $T \to 0$, which concludes the proof.
\end{proof}

\textbf{Proof of Proposition \ref{e:square_root_X_explicit formula}}. The proof of Proposition \ref{e:square_root_X_explicit formula}
is essentially based on the following Lemma.

\begin{lemma} \label{l:asympt_integral}
For every $\rho \in (-1,1)$ and $\varphi \maj 0$, the following asymptotics holds
\be \label{e:asympt_integral}
I(\theta, \varphi, \rho)
\sim \frac 1 {\sqrt{2 A_\theta}}
= \frac{\theta \varphi} { 2 \sqrt{ 1 + \frac{1+\rho^2}{16} \theta^2 \varphi^2} },
\qquad
\mbox{as } \theta \to 0.
\ee
\end{lemma}
\begin{proof}
It is not difficult to see that, for every $\delta \maj 0$, $\int_{|y| \ge \delta} e^{- A_{\theta} y^2} \dd y = o(e^{-A_{\theta} \, \delta^2}) = o\Bigl(e^{-\frac{\delta^2}{\varphi^2} \frac 1{\theta^2 }} \Bigr) = o(\theta)$ as $\theta \to 0$, so that
\be \label{e:first_integral_asympt}
\int_{|y| \mino \delta} e^{- A_{\theta} y^2} \frac{\dd y}{\sqrt{2 \pi}}
\sim
\int_{y \in \R} e^{- A_{\theta} y^2} \frac{\dd y}{\sqrt{2 \pi}}
= \frac 1 {\sqrt{2 A_\theta}}
\sim \frac{ \theta \varphi}2.
%\qquad
%\mbox{as } \theta \to 0.
\ee
The proof of \eqref{e:asympt_integral} implements the same idea for the integral $I(\theta, \varphi, \rho)$.
Let $\delta \maj 0$ and denote $f_\theta(y)  = e^{- \frac 14 \rho \, y  \sqrt{y^2 + \theta}}$, so that 
\begin{multline} \label{e:decomposition}
\sqrt{2 \pi} \, I(\theta, \varphi, \rho) =
\int_{\R} e^{
- A_{\theta} y^2} f_\theta(y)  \dd y
\\
= 
\int_{|y| \mino \delta} e^{
- A_{\theta} y^2} f_\theta(y)  \dd y
+
\int_{{|y| \ge \delta}} e^{
- A_{\theta} y^2} f_\theta(y)  \dd y
:= I_1(\theta) + I_2(\theta) \,.
\end{multline}
Assume without loss of generality $\theta \mino \delta^2$. 
Then, $|y| \ge \delta$ implies $ \sqrt{y^2 + \theta} \le \sqrt 2 |y|$, hence $f_\theta(y) \le e^{\frac{|\rho|}4 \sqrt 2 \, y^2} =  e^{c \, y^2}$ with $c = \frac {|\rho|}4 \sqrt 2$.
Consequently, for $\theta$ small enough (precisely: such that $A_{\theta} - c \maj 0$), we have
\be \label{e:estim_I_1}
\begin{aligned}
I_2(\theta)
%= \int_{{|y| \ge \delta}} e^{ - A_{\theta} y^2} f_\theta(y)  \dd y
&\le 
\int_{{|y| \ge \delta}} e^{
- (A_{\theta} - c) y^2}  \dd y
=
e^{- (A_{\theta} - c) \delta^2} 
\int_{y \ge \delta} 
\frac {2y} y \, e^{
- (A_{\theta} - c) (y^2 - \delta^2)}  \dd y
\\
&\le 
e^{- (A_{\theta} - c) \delta^2} 
\frac 1{\delta (A_{\theta} - c)}
\left[ e^{- (A_{\theta} - c) (y^2 - \delta^2)} \right]_{y=\infty}^{y=\delta}
\\
&= 
\frac{ e^{c \, \delta^2} } {\delta (A_{\theta} - c)}
e^{- A_{\theta} \delta^2}
= o_{\theta \to 0} (e^{- A_{\theta}  \delta^2})
= o_{\theta \to 0} \Bigl(e^{-\frac{\delta^2}{\varphi^2} \frac 1{\theta^2 }} \Bigr).
\end{aligned}
\ee 

Let us now estimate $I_1(\theta)$. 
For $|y| \mino \delta$, we have $f_\theta(y) \le e^{\frac{|\rho|}4 \delta \sqrt{\delta^2 + \theta}} \le e^{\frac{|\rho|}4 \sqrt 2 \, \delta^2} = e^{c \, \delta^2}$,
hence
\[
\begin{aligned}
|f_\theta(y) - 1|
&\le
\left|\int_0^y f_\theta'(z) \dd z \right|
=
\left| \int_0^y f_\theta(z) \frac{\dd}{\dd z} \Bigl( - \frac 14 \rho \, z  \sqrt{z^2 + \theta} \Bigr) \dd z \right| 
\\
&\le
e^{c \, \delta^2} \frac{|\rho|} 4 \left| \int_{\min(y,0)}^{\max(y,0)}  \frac{\dd}{\dd z} ( z  \sqrt{z^2 + \theta} ) \dd z \right| 
\le
e^{c \, \delta^2} \frac{|\rho|} 4 \delta \sqrt{\delta^2 + \theta}
\le
e^{c \, \delta^2} c \, \delta^2
\end{aligned}
\]
It follows that  $1 - e^{c \, \delta^2} c \, \delta^2 \le f_\theta(y) \le 1 + e^{c \, \delta^2} c \, \delta^2$ for every $y$ with $|y|\mino \delta$, therefore
\[
\begin{aligned}
(1 - e^{c \, \delta^2} c \,\delta)
\int_{|y| \mino \delta}
e^{ - A_{\theta} y^2}  \dd y
\le I_1(\theta) \le
(1 + e^{c \, \delta^2} c \, \delta)
\int_{|y| \mino \delta} e^{
- A_{\theta} y^2} \dd y \,,
\end{aligned}
\]
hence, in light of \eqref{e:first_integral_asympt},
\[
\begin{aligned}
(1 - e^{c \, \delta^2} c \,\delta)
\le 
\liminf_{\theta \to 0}
\frac{I_1(\theta)}{\int_{y \in \R} e^{- A_{\theta} y^2} \dd y}
\le 
\limsup_{\theta \to 0}
\frac{I_1(\theta)}{\int_{y \in \R} e^{- A_{\theta} y^2} \dd y}
\le
(1 + e^{c \, \delta^2} c \, \delta).
\end{aligned}
\]
Since $\delta \maj 0$ is arbitrary, we conclude that $I_1(\theta) \sim \int_{y \in \R} e^{- A_{\theta} y^2} \dd y \sim \frac 1 {\sqrt{2 A_\theta}} \sim \frac{ \theta \varphi}2$ as $\theta \to 0$, therefore \eqref{e:asympt_integral} follows from \eqref{e:decomposition} and \eqref{e:estim_I_1}.
\end{proof}

\begin{proof}[Proof of Proposition \ref{e:square_root_X_explicit formula}]
The expansion \eqref{e:asymptotics_expectation} now follows from equation \eqref{e:formula_sqrt_1} and Lemma \ref{l:asympt_integral}.
\end{proof}

\end{document}